\documentclass[12pt,draftcls,onecolumn]{ieeetran}
\usepackage{color}
\usepackage{graphicx}
\usepackage{amsmath}
\usepackage{times}
\usepackage{cite}
\usepackage{amssymb,dsfont}
\usepackage{geometry}
\renewcommand{\baselinestretch}{1.7}

\def\diag{{\rm diag}\,}

\def\Exp{{\mathbb{E}}\,}
\def\tr{{\rm tr}\,}

\def\diag{{\rm diag}\,}
\def\real{{\mathfrak{R} }\,}
\def\imag{{\mathfrak{I}}\,}

\def\be{\begin{equation}}
\def\ee{\end{equation}}
\def\ba{\left[\begin{array}}
\def\ea{\end{array}\right]}
\def\bea{\begin{eqnarray}}
\def\eea{\end{eqnarray}}

\newcommand{\mb}[1]{\mathbf{#1}}
\newcommand{\mr}[1]{\mathrm{#1}}
\newcommand{\mc}[1]{\mathcal{#1}}

\def\ba{{\bf a}}

\def\t{\mr{T}}

\newtheorem{theorem}{\textbf{Theorem}}
\newtheorem{proposition}{\textbf{Proposition}}

\newtheorem{definition}{\textbf{Definition}}

\begin{document}
\title{Multi-Antenna System Design with Bright Transmitters and Blind
Receivers}
%\vspace{-1.5cm}}
%
% Single address.
% ---------------
\date{}
\renewcommand{\baselinestretch}{1.7}
\IEEEoverridecommandlockouts
\author{Liangbin Li, Hamid Jafarkhani\thanks{This work was supported in part by the NSF Award
CCF-0963925. Part of this work was presented at IEEE Globecom 2011.}
\\Center for
Pervasive Communications \& Computing, University of California,
Irvine} \maketitle

\vspace{-10mm}
\begin{abstract}
This paper considers a scenario for multi-input multi-output (MIMO)
communication systems when perfect channel state information at the
transmitter (CSIT) is given while the equivalent channel state
information at the receiver (CSIR) is not available. Such an
assumption is valid for the downlink multi-user MIMO systems with
linear precoders that depend on channels to all receivers. We
propose a concept called \emph{dual systems with zero-forcing
designs} based on the duality principle, originally proposed to
relate Gaussian multi-access channels (MACs) and Gaussian broadcast
channels (BCs). For the two-user $N\times 2$ MIMO BC with $N$
antennas at the transmitter and two antennas at each of the
receivers, we design a \emph{downlink interference cancellation (IC)
transmission scheme} using the dual of uplink MAC systems employing
IC methods. The transmitter simultaneously sends two precoded
Alamouti codes, one for each user. Each receiver can zero-force the
unintended user's Alamouti codes and decouple its own data streams
using two simple linear operations independent of CSIR. Analysis
shows that the proposed scheme achieves a diversity gain of $2(N-1)$
for equal energy constellations with short-term power and rate
constraints. Power allocation between two users can also be
performed, and it improves the array gain but not the diversity
gain. Numerical results demonstrate that the bit error rate of the
downlink IC scheme has a substantial gain compared to the block
diagonalization method, which requires global channel
information at each node. %It also outperforms an opportunistic
%time division multi-access (TDMA) scheme, that uses the
%proposed \emph{dual Alamouti codes} for orthogonal
%transmission of each user's messages with stronger channel
%norms, in the low and median signal-to-noise ratio range for
%high-rate transmission.
%We also propose an opportunistic
%time division multi-access (TDMA) scheme that uses the \emph{dual
%Alamouti codes} for orthogonal transmission of each user's messages.
%In the Monte-Carlo simulation, the downlink IC scheme achieves
%better BER in the low and median signal-to-noise (SNR) range for
%high-rate transmission, while the opportunistic TDMA outperforms in
%the high SNR range or low-rate transmission.
\end{abstract}


\renewcommand{\baselinestretch}{2}

\section{Introduction}
System performance can be improved through learning the fading
coefficients of wireless channels. A pilot sequence is
inserted at the beginning of each data stream to help the
receiver estimate the unknown fading coefficients. Then,
techniques such as receive beamforming or coherent detection
can be conducted to exploit the known channel state
information at the receiver (CSIR) before it is outdated.
Channel state information at the transmitter (CSIT) can also
be obtained through feedback channels if the channel coherent
interval is longer than the feedback delay. Also, when the
system is operating in the time division duplex (TDD) mode,
the forward and the reverse channel coefficients are
approximately the same due to reciprocity\cite{MaHo06}. Then,
the CSIR obtained at the reverse channel is used as the CSIT
for the forward channel. Techniques such as rate adaption,
power allocation, and transmit beamforming can be performed at
the transmitter using the knowledge of CSIT. With respect to
the assumptions on the channel information, communication
systems can be classified into four categories as listed in
Table~\ref{table}. The first three categories have been
extensively discussed in \cite{hj}, while to the best of our
knowledge, communication systems with CSIT and no CSIR have
been ignored. Generally, obtaining channel information at the
receiver is easier than obtaining it at the transmitter. The
use of System D, although not seemingly natural, can be
illustrated for transmission in broadcast channels (BCs).

%\begin{enumerate}
%  \item \emph{Example A}: consider a point-to-point multi-input multi-output (MIMO)
%  system operating on the TDD mode. In a complete communication session with downlink and uplink
%transmission, the system individually learns both the forward and
%the reverse channels. Then, two pilot sequences are needed: one sent
%by the base station (BS) to learn the forward channel, and one sent
%by the mobile station (MS) to learn the reverse channel. Due to the
%TDD assumption, the same channel coefficients are learned twice in a
%complete communication session by individual training. On the other
%hand, if pilot sequence is only sent by the MS but not the BS, the
%forward system obtains CSIR and no CSIT, and the reverse system has
%equivalent CSIT and no CSIR, that falls into a System D. With a
%transmission scheme for System D, the MS can operate without knowing
%channel information and simplifying the node designs. Also, it helps
%the BS to save training power.
%  \item
%\end{enumerate}

In a multi-user multi-input multi-output (MIMO) BC, the transmitter
simultaneously sends multiple independent data streams for each
user. The information theoretical aspects of the BC, e.g., the sum
capacity or the capacity region of the vector MIMO Gaussian BC, have
received much attention\cite{ViTse03, ViJinGo03,WeiStSh06}. The
approach to achieve the capacity is through exploiting CSIT using a
nonlinear precoding method called dirty paper coding
(DPC)\cite{Costa83}. The knowledge of the codewords of the
previously encoded data streams is used to encode a new data stream.
For practical systems, such nonlinear precoders are discouraged
because of its high complexity. A class of linear precoders, known
as the zero-forcing (ZF) precoder or the block diagonalization (BD)
method, is proposed to reduce the complexity\cite{SpSwHa04}. It is
also known that the low-complexity linear alternative achieves the
maximum multiplexing gain\cite{ShCh07}. Data streams of each user
are independently encoded and a ZF precoder is designed for each
data stream to null out its interference at unintended receivers.
The design of ZF precoders is later generalized in \cite{Sung09}
using minimum mean square error criterions. One limitation of the ZF
precoding methods is the requirement of perfect CSIT. Also, the
number of allowable users is constrained by the available antennas
at the base station. For a cell with a large number of receivers,
multi-user scheduling and quantized feedback design are discussed in
\cite{Yoo07,Kh11}. The achievable sum-capacity scales linearly with
$\log \log K$, where $K$ denotes the number of receivers.

There is another practical limitation of the ZF precoders: its
implementation incurs substantial exchange of channel
information. The transmitter can learn the global CSIT through
the feedback channels or reciprocal channels, and design the
ZF precoders. Note that each user's precoder also depends on
the channels to the other users and each receiver does not
know the other users' channels. The equivalent channels, that
consist of both the precoders and the fading channels, are
still unknown at each receiver. Another round of training to
learn the equivalent channels is necessary for coherent
detection. Otherwise, the ZF precoders decouple the BC into
parallel non-interfering systems with CSIT and no CSIR. A
transmission scheme for System D can be applied directly for
each system without further training. This motivates our work
to study transmission with CSIT and no CSIR.

Most previous designs of ZF precoders are focused on maximizing the
sum-capacity or the multiplexing gain, which needs infinite
dimensions for signaling. From the reliability perspective, the
diversity gain using finite signaling dimensions is of equal
importance. A comparison between the opportunistic scheduling and
the BD transmission using space-time block codes (STBCs) is made in
\cite{Lee10}. For a $K$-user $N\times M$ MIMO BC where the
transmitter has $N$ antennas and each of the $K$ receivers has $M$
antennas, allowing only the user with the strongest channel gains to
transmit can achieve the maximum diversity gain of $KMN$. For the BD
transmission concatenated with STBCs, the achievable diversity gain
is $(N-(K-1)M)M$, which is far less than the maximum possible
diversity gain. In this paper, we also aim at improving the
achievable diversity gains for ZF precoders.

To achieve these goals, we adopt the idea of duality, that was
originally proposed to relate the Gaussian multi-access
channel (MAC) and the Gaussian BC with perfect channel
information at all nodes \cite{ViTse03,ViJinGo03}. A new
duality principle is proposed to connect two linear systems
with ZF designs. We apply this principle to transform known
transmission schemes in original systems into new schemes for
their dual systems. We use an uplink MAC system performing
interference cancellation (IC) techniques \cite{NaSeCa,AlCa}
as the original system, and propose downlink IC schemes for
BCs. The contribution of this paper is summarized as follows:
\begin{enumerate}
  \item We propose a duality principle to connect two
real systems with linear ZF designs. The transmit (receive) filters
in the first system are used as the receive (transmit) filters in
the second system. When both systems have the same input power
distribution and the channel matrix of the first system is transpose
of that of the second system, the instantaneous signal-to-noise
ratios (SNRs) at the outputs of the receive filters are also
identical for both systems. We illustrate this principle by
constructing the dual of Alamouti systems, and obtain a new
beamforming method called \emph{dual Alamouti codes}.
  \item For the two-user $N\times 2$ MIMO BC, we
  propose \emph{a downlink IC scheme} that concurrently sends precoded Alamouti codes for each user, and cancels interference and decodes messages blindly at each receiver.
  %Two precoded Alamouti codes, one for each user, are concurrently transmitted. Each receiver can zero-force unintended signals and decouple its own data streams using two linear operations independent of CSIR.
  It achieves a diversity gain of $2(N-1)$, rate-one
  per user, and symbol-by-symbol decoding complexity. Compared to the
  linear BD methods, although it does not achieve the maximum rate for $N>2$, the proposed scheme requires less number of transmit antennas, does not need perfect global CSIR at
  each receiver, and achieves higher diversity gain that improves the bit error rate (BER) performance. Also, the proposed scheme shows superiority in terms of diversity gain compared to the BD method even if it is concatenated with a STBC\cite{ChAnHe04,Lee10}. %Also compared to
%  the opportunistic time-division multi-access (TDMA) scheme where each user's symbols are transmitted using dual
%  Alamouti codes and the user with stronger channel norms is scheduled to transmit, the downlink IC scheme has higher rate and achieves better BER
%  for high-rate transmission in the low and median SNR ranges.
\end{enumerate}

The rest of the paper is organized as follows. In Section
\ref{sec-duality}, we present our definition of duality, and
construct the dual Alamouti codes from the original Alamouti
codes for point-to-point MIMO systems. In Section
\ref{sec-downlinkIC}, we propose the downlink IC scheme for
the two-user $N\times 2$ MIMO BC, and analyze its diversity
gain performance. Simulation results are provided in Section
\ref{sec-numer}. Finally, concluding remarks are given in
Section \ref{sec-conclusion} and involved proofs are included
in the appendices.

\textit{Notation}: For matrix $\mb{A}$, let $\mb{A}^\t$, $\mb{A}^*$,
and $\|\mb{A}\|$ denote its transpose, Hermitian, and Frobenius
norm, respectively. Denote $\mb{0}_{m, n}$, $\mb{0}_m$, and
$\mb{I}_m$ as an $m\times n$ all-zero matrix, an $m\times m$
all-zero matrix, and an $m\times m$ identity matrix, respectively.
The Kronecker product is denoted as $\otimes$. We define
$\mc{N}(0,1)$ and $\mc{CN}(0,1)$ as real Gaussian distribution and
circularly symmetrical complex Gaussian distribution with zero mean
and unit variance, respectively. We also use $\real a$ and $\imag a$
to denote the real and imaginary components of a complex variable
$a$, respectively. The sets of real and complex numbers are denoted
as $\mathds{R}$ and $\mathds{C}$, respectively. Similarly, the sets
of $N\times M$ real matrix and complex matrix are denoted as
$\mathds{R}^{N\times M}$ and $\mathds{C}^{N\times M}$, respectively.

\section{Duality for ZF Designs}\label{sec-duality}
We design systems with CSIT and no CSIR from known systems with CSIR
and no CSIT. These two systems are connected through duality where
the roles of transmit and receive filters are exchanged. In this
section, we present a new duality principle for linear systems with
ZF designs. Subsection \ref{subsec-principle} discusses the duality
principle, and Subsection \ref{subsec-Alamouti} provides an example
to construct the dual of the Alamouti systems.

\subsection{The Duality Principle}\label{subsec-principle}
We consider two real systems as illustrated in Fig.~\ref{fig-dual}.
System 1 has an input vector $\mb{s}\in \mathds{R}^{1\times n}$ and
an output vector $\mb{d} \in \mathds{R}^{1\times m}$ with
input-output relationship
 \begin{align}\label{eq-original11}\mb{d}=\mb{s}\mb{Z}+\mb{n},\end{align}
where $\mb{Z}\in \mathds{R}^{n\times m}$ denotes the
equivalent channel matrix, and $\mb{n}\in \mathds{R}^{1\times
m}$ denotes the i.~i.~d. $\mc{N}(0,1)$ additive white Gaussian
noise (AWGN) vector. System 2 has an input vector $\mb{x}\in
\mathds{R}^{1\times m}$ and an output vector $\mb{r}\in
\mathds{R}^{1\times n}$ with the input-output relationship
\be\mb{r}=\mb{x}\mb{F}+\mb{w},\ee where $\mb{F}\in
\mathds{R}^{m\times n}$ and $\mb{w}\in \mathds{R}^{1\times n}$
denote the equivalent channel matrix and the i.~i.~d.
$\mc{N}(0,1)$ noise vector, respectively.

Multiple information streams $c_l, (l=1,\ldots,J)$ are
transmitted simultaneously through these two systems using
transmit and receive filters. For System 1, Stream $c_l$ is
sent using $\mb{u}_l\in \mathds{R}^{1\times n}$ as the
transmit filter and $\mb{v}_l\in \mathds{R}^{1\times m}$ as
the receive filter. The filters are normalized as
{\small$\|\mb{u}_l\|^2=\|\mb{v}_l\|^2=1$.} The input vector
$\mb{s}$ is generated through linear superposition of all
streams as
{\small$\mb{s}=\underset{l=1:J}{\sum}\sqrt{\mc{P}_l}\mb{u}_lc_l$},
where the coefficient $\mc{P}_l$ denotes power allocation
profile for Stream $c_l$ and satisfies the power constraint
$\underset{l=1:J}{\sum}\mc{P}_l=\mc{P}$. The receiver
calculates $\hat{y}_k=\mb{d}\mb{v}_k^\t$ to extract Stream
$c_k$. From \eqref{eq-original11}, the equivalent system at
the output of $\mb{v}_k$ can be expressed as
\be\label{eq-original1}
\hat{y}_k=\left(\sum_{l=1:J}\sqrt{\mc{P}_l}\mb{u}_lc_l\right)\mb{Z}\mb{v}_k^\t+\mb{n}\mb{v}_k^\t.\ee
The filters are designed based on ZF constraints. In other
words, the output of $\mb{v}_k$ contains no component of the
stream $\mb{c}_l$ for $l\neq k$. Mathematically, we have
\be\label{eq-ZF}\mb{u}_l\mb{Z}\mb{v}_k^\t=0, l\neq k.\ee Thus,
the equivalent system in \eqref{eq-original1} can be
simplified as
$\hat{y}_k=\sqrt{\mc{P}_k}c_k\mb{u}_k\mb{Z}\mb{v}_k^\t+\mb{n}\mb{v}_k^\t.$
The SNR at the output of $\mb{v}_k$ can be expressed as
\be\label{eq-SNRorig}\mr{SNR}_{\mr{orig},k}=\mc{P}_k\|\mb{u}_k\mb{Z}\mb{v}_k^\t\|^2.\ee

For System 2, Stream $c_l$ is sent using the transmit filter
$\mb{v}_l$ (the receive filter of System 1) and the receive filter
$\mb{u}_l$ (the transmit filter of System 1). The input vector is
generated using the same power allocation profile as System 1 by
{\small$\mb{x}=\underset{l=1:J}{\sum}\sqrt{\mc{P}_l}\mb{v}_lc_l$}.
The receiver uses the output at the receive filter $\mb{u}_k^\t$ to
extract Stream $c_k$. The system equation can be expressed as
\be\label{eq-dual}
\hat{r}_k=\left(\sum_{l=1:J}\sqrt{\mc{P}_l}\mb{v}_lc_l\right)\mb{F}\mb{u}_k^\t+\mb{w}\mb{u}_k^\t.
\ee\hspace{-3pt} The receiver treats interfering streams as noises,
and the equivalent output SNR at the filter $\mb{u}_k$ is
\be\label{eq-SNRdual}\mr{SNR}_{\mr{dual},k}=\frac{\mc{P}_k\|\mb{v}_k\mb{F}\mb{u}_k^\t\|^2}{\underset{l\neq
k}{\sum}\mc{P}_l\|\mb{v}_l\mb{F}\mb{u}_k^\t\|^2+1}.\ee The first
term of the denominator represents the power of interference, while
the second term represents the variance of the equivalent noise
$\mb{w}\mb{u}_k^\t$. In what follows, we define duality.
\begin{definition}\label{def-1} System 2 is called
the \emph{dual} of System 1 with ZF designs if:
\begin{enumerate}
  \item System 1 uses $\mb{u}_l$ and $\mb{v}_l$ as
transmit and receive filters, respectively. System 2 uses
$\mb{v}_l$ and $\mb{u}_l$ as transmit and receive filters,
respectively.
\item Both systems have the same power allocation profile $\mc{P}_l$.
  \item The channel matrices are related by $\mb{F}=\mb{Z}^\t$.
  \item The filters $\mb{u}_l$ and $\mb{v}_l$ satisfy the ZF relationship in
  \eqref{eq-ZF}.
\end{enumerate}
\end{definition}

\begin{proposition}\label{prop-1}
Both the \emph{original} and the \emph{dual} systems have the same
SNR at the output of the $k$th receive filter.
\end{proposition}
\begin{proof}
When $\mb{F}=\mb{Z}^\t$, due to the ZF relationship in
\eqref{eq-ZF}, the power of interference in \eqref{eq-SNRdual}
is $\mc{P}_l\|\mb{v}_l\mb{Z}^\t\mb{u}_k^\t\|^2=0$ for $l\neq
k$. Thus, \eqref{eq-SNRdual} can be simplified as
$\mr{SNR}_{\mr{dual},k}=\mc{P}_k\|\mb{v}_k\mb{Z}^\t\mb{u}_k^\t\|^2$.
From \eqref{eq-SNRorig}, we have
$\mr{SNR}_{\mr{orig},k}=\mr{SNR}_{\mr{dual},k}$.
\end{proof}

\subsection{Application in point-to-point MIMO systems: Dual Alamouti codes}\label{subsec-Alamouti}
In this subsection, we present an example of applying the
duality principle to obtain the dual of Alamouti
systems\cite{Alamouti98}. For a $2\times N$ MIMO system with
two antennas at the transmitter and $N$ antennas at the
receiver, the Alamouti system sends two symbols $s_1$ and
$s_2$ in two time slots. The symbols are drawn independently
and uniformly from a normalized constellation $\mc{S}$ with
finite cardinality. For simplicity, we assume $\mc{S}$ to be a
PSK constellation. The transmitted matrix has an Alamouti
structure
\begin{align}
\mb{S}=\sqrt{\frac{P}{2}}\left[\begin{array}{cc}s_1& s_2\\
-s_2^*& s_1^*\end{array}\right],
\end{align}
where $P$ denotes the transmitted power per time slot. Let us
denote $y_{ti}$ as the receive signal at time slot $t$ and
receive Antenna $i$. The receiver calculates the negative
conjugate of the receive signal in time slot 2 and stacks the
signals into a $2N\times 1$ vector
$\tilde{\mb{y}}=\left[y_{11}\ -y_{21}^*\ \cdots \ y_{1N}\
-y_{2N}^*\right]^\t$. The receiver further decouples these two
symbols by calculating
$\hat{\mb{y}}=\frac{\tilde{\mb{G}}^*}{\|\mb{G}\|}\tilde{\mb{y}}$.
The matrix $\mb{G}\in \mathds{C}^{2\times N}$ denotes the
channel matrix whose $(j,i)$ entry $g_{ji}$ denotes the
channel coefficient from transmit Antenna $j$ to receive
Antenna $i$. The matrix $\tilde{\mb{G}}\in
\mathds{C}^{2N\times 2}$ is composed of the $2\times 2$
equivalent Alamouti channel matrices $\tilde{\mb{G}}_i$ at
Antenna $i$,
\begin{align}\label{eq-alarec}
\tilde{\mb{G}}_i=\left[\begin{array}{cc}g_{1i}& g_{2i}\\ -g_{2i}^*&
g_{1i}^*\end{array}\right],
\tilde{\mb{G}}=\left[\begin{array}{ccc}\tilde{\mb{G}}_1^*&
\cdots&\tilde{\mb{G}}_N^*
\end{array}\right]^*.
\end{align}
Note that $\hat{\mb{y}}\in \mathds{C}^{2\times 1}$.
Symbol-by-symbol decoding of $s_j$ is performed based on the
$j$th component of $\hat{\mb{y}}$.

In what follows, we use duality to obtain the dual of Alamouti
systems. Since symbol conjugates are excluded from the
definition of duality and Alamouti codes send conjugates of
symbols, we separate the real and imaginary components of one
symbol and stack them into a $2\times 1$ vector. Then, the
conjugate of this symbol can be equivalently written as a
linear transformation of this vector using a $2\times 2$
matrix $\diag [1,-1]$. For this reason, we also need to
separate the real and imaginary parts of the receive signals.
Let us denote the $j$th component of $\hat{\mb{y}}$ as
$\hat{{y}}_j$, which is the output after symbols are decoupled
at the receiver. We stack the real and imaginary parts of
$\hat{{y}}_j$ into a $1\times 4$ vector $\left[\real
\hat{y}_{1}\ \imag \hat{y}_{1} \ \real \hat{y}_{2}\ \imag
\hat{y}_{2}\right]$. For the Alamouti system, the equivalent
equations after symbol decoupling can be expanded as
\begin{align}\label{eq-equi}
\left[\begin{array}{c}\real \hat{y}_{1}\\ \imag \hat{y}_{1} \\
\real \hat{y}_{2}\\ \imag
\hat{y}_{2}\end{array}\right]^\t=\sqrt{\frac{P}{2}}\left[\begin{array}{c}\real s_1 \\
\imag s_1
\\
\real s_2 \\ \imag
s_2\end{array}\right]^\t\mb{U}\underset{\mb{I}_2\otimes \hat{\mb{G}}}{\underbrace{\left[\begin{array}{ll}\hat{\mb{G}}&\mb{0}_{4, 2N}\\
\mb{0}_{4,
2N}&\hat{\mb{G}}\end{array}\right]}}\mb{V}^\t+{\underset{\bar{\mb{n}}}{\underbrace{\left[\begin{array}{c}
 \bar{\mb{n}}_{1}\\
 \bar{\mb{n}}_{2}\end{array}\right]^\t}}}\mb{V}^\t,\end{align}
where $\mb{U}\in \mathds{R}^{4\times 8}$, $\hat{\mb{G}}\in
\mathds{R}^{4\times 2N}$, $\mb{V}\in \mathds{R}^{4\times 4N}$,
and $\bar{\mb{n}}\in \mathds{R}^{1\times 4N}$ denote the
transmit processing matrix, the equivalent fading channel
matrix, the receive processing matrix, and equivalent noise
vector, respectively. The matrices $\hat{\mb{G}}$ and
$\bar{\mb{n}}_t$ can be expressed as
\begin{align*}
\hat{\mb{G}}=\left[\begin{array}{rrrrr}\real g_{11}& \imag g_{11}& \cdots & \real g_{1N} & \imag g_{1N}\\
-\imag g_{11}& \real g_{11} &\cdots& -\imag g_{1N} & \real
g_{1N}\\
\real g_{21}& \imag g_{21}& \cdots & \real g_{2N} & \imag g_{2N}\\
-\imag g_{21}& \real g_{21} &\cdots& -\imag g_{2N} & \real
g_{2N}\end{array}\right],
\bar{\mb{n}}_t=\left[\begin{array}{ccccc}\real n_{t1}&\imag n_{t1}&
\cdots &\real n_{tN}& \imag n_{tN}\end{array}\right]^\t,
\end{align*}
where $n_{ti}$ denotes the $\mc{CN}(0,1)$ noise at time slot
$t$ and receive Antenna $i$. The transmit processing matrix
{\setlength{\arraycolsep}{1pt}$\mb{U}=\left[\begin{array}{cc}\mb{I}_4&\mb{A}\end{array}\right]$}
and the receive processing matrix
$\mb{V}=\frac{(\mb{D}\mc{G})^\t}{\|\mb{G}\|}$ with
{\setlength{\arraycolsep}{1pt}$\mb{D}=\diag
\left[\mb{I}_{2N},\mb{I}_N\otimes\diag[-1,1]\right]$} and
{\setlength{\arraycolsep}{1pt}$\mc{G}=\left[\begin{array}{cc}\hat{\mb{G}}&\bar{\mb{G}}
\end{array}\right]^\t$}. The matrices $\mb{A}$ and $\bar{\mb{G}}$ can be expressed as
\begin{align}\label{eq-matrices}\mb{A}=\left[\begin{array}{rrrrrrrr}0&0&1&0\\0&0&0&-1\\-1&0&0&0\\0&1&0&0\end{array}\right], \bar{\mb{G}}=\left[\begin{array}{rrrrr}-\real g_{21}& \imag g_{21}& \cdots & -\real g_{2N} & \imag g_{2N}\\
-\imag g_{21}& -\real g_{21} &\cdots& -\imag g_{2N} & -\real
g_{2N}\\
\real g_{11}& -\imag g_{11}& \cdots & \real g_{1N} & -\imag g_{1N}\\
\imag g_{11}& \real g_{11} & \cdots & \imag g_{1N} & \real
g_{1N}\end{array}\right].
\end{align}
For the transmit processing matrix $\mb{U}$, each row
represents one dispersion matrix of Alamouti codes\cite{hj}.
Inside the receive processing matrix $\mb{V}$, the matrix
$\mb{D}$ contains the calculations of the negative conjugate
of the receive signal in time slot 2, while the matrix
$\frac{\mc{G}}{\|\mb{G}\|}$ decouples these two symbols.

Note that the system equation in \eqref{eq-equi} resembles the
\emph{original} system with a $1\times 8$ input vector, a
$4N\times 1$ output vector, and a channel matrix
$\mb{I}_2\otimes \hat{\mb{G}}$. Four data streams are
transmitted: $c_1=\real s_1$, $c_2=\imag s_1$, $c_3=\real
s_2$, and $c_4=\imag s_2$. The streams are sent using the rows
of $\frac{\mb{U}}{\sqrt{2}}$ as the transmit filters (the
transmit filters are normalized by dividing by $\sqrt{2}$) and
the rows of $\mb{V}$ as the receive filters. Let $n=8$,
$m=4N$, and $\mb{Z}=\mb{I}_2\otimes \hat{\mb{G}}$. The
transmit power is $\mc{P}=4P$ per transmission, and equally
allocated among four streams. Also, it can be verified that
these filters satisfy the ZF conditions in \eqref{eq-ZF}.
Using the definition of duality, we obtain the system equation
for its \emph{dual} system as
{\setlength{\arraycolsep}{1pt}\begin{align}\label{eq-aladual} \left[\begin{array}{c}\real \hat{r}_1 \\
\imag \hat{r}_1
\\ \real \hat{r}_2 \\ \imag
\hat{r}_2\end{array}\right]^\t=\sqrt{P}\left[\begin{array}{c}\real s_1 \\
\imag s_1
\\
\real s_2 \\ \imag
s_2\end{array}\right]^\t\mb{V}\left[\begin{array}{ll}\hat{\mb{G}}^t&\mb{0}_{2N, 4}\\
\mb{0}_{2N,
4}&\hat{\mb{G}}^t\end{array}\right]\frac{\mb{U}^\t}{\sqrt{2}}+\mb{w}\frac{\mb{U}^\t}{\sqrt{2}},\end{align}}
\hspace{-3pt}where the rows of $\mb{V}$ are the transmit
filters and the rows of $\frac{\mb{U}}{\sqrt{2}}$ are the
receive filters. Note that $\mb{V}$ depends on the channel
information, while $\mb{U}$ is independent of the channel
information. The dual system uses CSIT and does not require
CSIR. For the dual system, denote the channel path from
transmit Antenna $i$ to receive Antenna $j$ as $h_{ij}$. From
the definition of duality, we have $g_{ji}=h_{ij}^*$.
Replacing $g_{ji}$ with $h_{ij}^*$ in \eqref{eq-aladual},
combining the real and imaginary components, and converting
back to matrix expression, we obtain the \emph{dual Alamouti
codes}.

More specifically, the dual Alamouti system is described as
follows. The transmitter has $N$ antennas and the receiver has
$2$ antennas. Denote the $N\times 2$ channel matrix as
$\mb{H}$ whose $(i,j)$ entry is $h_{ij}$. Two symbols $s_1$
and $s_2$ first form a $2\times 2$ Alamouti block code. Then,
a scaled Hermitian of $\mb{H}$ is used as the precoder. The
$2\times N$ transmitted matrix can be expressed as
\begin{align}\label{eq-dualcodes}
\mb{X}=\sqrt{P}{\left[\begin{array}{cc}s_1 & s_2
\\ -s_2^* &
s_1^*\end{array}\right]}\frac{\mb{H}^*}{\|\mb{H}\|}.\end{align}
The transmit power can be verified to be $P$ per time slot,
and powers are equally allocated between the two symbols. The
$(t,i)$ entry of $\mb{X}$ is sent at Antenna $i$ and time slot
$t$. The receive signal at
Antenna $j$ can be expressed as {\setlength{\arraycolsep}{1pt}\begin{align}\left[\begin{array}{c}r_{1j}\\
r_{2j}\end{array}\right]=\mb{X}\mb{h}_j+\left[\begin{array}{c}w_{1j}\\
w_{2j}\end{array}\right]=\frac{\sqrt{P}}{\|\mb{H}\|}\left[\begin{array}{cc}s_1
& s_2
\\ -s_2^* & s_1^*\end{array}\right]\mb{H}^*\mb{h}_j+\left[\begin{array}{c}w_{1j}\\
w_{2j}\end{array}\right]=\frac{\sqrt{P}}{\|\mb{H}\|}\left[\begin{array}{cc}s_1
& s_2
\\ -s_2^* & s_1^*\end{array}\right]\left[\begin{array}{c}\mb{h}_1^*\mb{h}_j\\ \mb{h}_2^*\mb{h}_j \end{array}\right]+\left[\begin{array}{c}w_{1j}\\
w_{2j}\end{array}\right], \label{eq-dualtran}\end{align}}where
$\mb{h}_j$ denotes the channel vector to receive Antenna $j$,
i.e., the $j$th column of $\mb{H}$, and $w_{tj}$ denotes the
$\mc{CN}(0,1)$ AWGN at receive Antenna $j$ and time slot $t$.
Take conjugate of the receive signals in time slot 2. We can
stack receive signals into a $4\times 1$ vector as
\begin{align}
\left[\begin{array}{c}r_{11}\\
r_{21}^*\\ r_{12}\\
r_{22}^*\end{array}\right]=\frac{\sqrt{P}}{\|\mb{H}\|}\left[\begin{array}{cc}\mb{h}_1^*\mb{h}_1& \mb{h}_2^*\mb{h}_1\\ \mb{h}_1^*\mb{h}_2 &-\mb{h}_1^*\mb{h}_1 \\\mb{h}_1^*\mb{h}_2& \mb{h}_2^*\mb{h}_2\\ \mb{h}_2^*\mb{h}_2 &-\mb{h}_2^*\mb{h}_1 \end{array}\right]\left[\begin{array}{c}s_1\\
s_2\end{array}\right]+\left[\begin{array}{c}w_{11}\\
w_{21}^*\\ w_{12}\\
w_{22}^*\end{array}\right].
\end{align}

Without CSIR, the receiver decouples $s_1$ and $s_2$ as,
{\small\setlength{\arraycolsep}{1pt}\begin{align}\label{eq-step1}
&\hat{r}_1=\frac{r_{11}+r_{22}^*}{\sqrt{2}}=\frac{\sqrt{P}}{\sqrt{2}\|\mb{H}\|}\left(s_1\mb{h}_1^*\mb{h}_1+s_2\mb{h}_2^*\mb{h}_1+(-s_2)\mb{h}_2^*\mb{h}_1+s_1\mb{h}_2^*\mb{h}_2\right)+\frac{w_{11}+w_{22}^*}{\sqrt{2}}=\sqrt{\frac{P}{2}}\|\mb{H}\|s_1+\frac{w_{11}+w_{22}^*}{\sqrt{2}},\\
\label{eq-step2}&\hat{r}_2=\frac{r_{12}-r_{21}^*}{\sqrt{2}}=\frac{\sqrt{P}}{\sqrt{2}\|\mb{H}\|}\left(s_1\mb{h}_1^*\mb{h}_2+s_2\mb{h}_2^*\mb{h}_2-(-s_2)\mb{h}_1^*\mb{h}_1-s_1\mb{h}_1^*\mb{h}_2\right)+\frac{w_{12}-w_{21}^*}{\sqrt{2}}=\sqrt{\frac{P}{2}}\|\mb{H}\|s_2+\frac{w_{12}-w_{21}^*}{\sqrt{2}}.
\end{align}}
The power of the equivalent noise is normalized by dividing by $\sqrt{2}$. %\hspace{-2pt}In \eqref{eq-step1} and \eqref{eq-step2}, the power of
To decode $s_j$, the receiver performs the maximum-likelihood (ML)
decoding for the equivalent systems in \eqref{eq-step1} and
\eqref{eq-step2} as{\setlength{\arraycolsep}{1pt}\begin{eqnarray}
\underset{s\in
\mc{S}}{\min}\left|\hat{r}_j-\sqrt{\frac{P}{2}}\|\mb{H}\|s\right|^2=\underset{s\in
\mc{S}}{\max}\ 2\real\left(\hat{r}_j
s^*\right)-\sqrt{\frac{P}{2}}\|\mb{H}\|\left|s\right|^2,
j\in\{1,2\}.\label{eq-ML}
\end{eqnarray}}
Since $\mc{S}$ is a PSK constellation, the second term
{\small$\sqrt{P/2}\|\mb{H}\|\left|s\right|^2$} has a constant value
for all points in $\mc{S}$ and thus can be
ignored\footnote{Generalizing Dual Alamouti codes for other
constellations such as QAM can be found in \cite{DualCode}}.
Eq.~\eqref{eq-ML} can be further simplified as
{\small$\underset{s\in \mc{S}}{\max}\ \real\left(\hat{r}_j
s^*\right)$}. This shows that the ML decoding can be performed
without CSIR. Symbol-by-symbol decoding complexity is also obtained,
i.e. $s_1$ and $s_2$ can be decoded separately.

The dual Alamouti codes need perfect CSIT and no CSIR, and are
belong to System D in Table \ref{table}. The performance of
dual Alamouti codes is described next. Two symbols are sent in
two time slots from \eqref{eq-dualcodes}. Thus, the rate of
the dual codes is one. From \eqref{eq-step1} and
\eqref{eq-step2}, the receive SNRs of both symbols are
$\sqrt{\frac{P}{2}}\|\mb{H}\|$, which is the same as that of
the original Alamouti systems. This confirms with Proposition
\ref{prop-1}. Since the ML decoding of dual Alamouti codes can
be conducted without CSIR, both Alamouti codes and its dual
codes achieve the same array gain and diversity gain.

\section{The Downlink IC Scheme for a Two-user MIMO
BC}\label{sec-downlinkIC}

This section focuses on the two-user MIMO BC system with CSIT and no
CSIR. Using our dual Alamouti code, each user can transmit in an
orthogonal time slot to avoid interference. However, this time
division multi-access (TDMA) method sacrifices the transmission
rate. We propose a concurrent transmission scheme for these systems
using the duality principle obtained in Section \ref{sec-duality}.
In the uplink two-user MAC, the IC scheme in \cite{NaSeCa,AlCa}
concurrently transmits both users' symbols through Alamouti codes
and linearly zero-forces the interference at the receiver. It
achieves full transmit diversity and user-by-user decoding
complexity. Since it satisfies the conditions for the original
system described in Definition \ref{def-1}, we use the uplink IC
scheme as the original system and construct its dual system for BCs.
In Subsection \ref{subsec-MAC}, we review the uplink IC scheme.
Subsection \ref{subsec-BC} presents the downlink IC scheme. Further
discussions on power allocation and diversity gain are provided in
Subsection \ref{subsec-dis}.

We describe our transmission schemes for a two-user $N\times 2$ BC
with $N$ antennas at the transmitter and two antennas at each
receiver. The generalization to a BC with any number of users and
antennas is straightforward using the duality principle. The channel
coefficients are i.~i.~d.~$\mc{CN}(0,1)$ distributed and known
perfectly at the transmitter but not at the receivers. Also, we
assume that the channels are block fading and are constant during
each transmission. We adopt the same notation as in Section
\ref{sec-duality} and add a superscript $k$ as user index. Thus,
$g_{ji}^{(k)}$ denotes the channel coefficients from transmit
Antenna $j$ of User $k$ to receive Antenna $i$ in the original MAC,
and $h_{ij}^{(k)}$ denotes the channel coefficient from transmit
Antenna $i$ to receive Antenna $j$ of User $k$ in the dual BC.

\subsection{The Uplink IC Scheme}\label{subsec-MAC}
We review the IC scheme for a two-user $2\times N$ MAC with two
antennas at each user and $N$ antennas at the
receiver\cite{NaSeCa,AlCa}. The system diagram is illustrated in
Fig.~\ref{fig-IC-uplink}. User $k$ sends two independent symbols
$s_1^{(k)}$ and $s_2^{(k)}$ to the receiver encoded by Alamouti
codes using power $\frac{P}{2}$ per time slot per user. Both users
transmit concurrently. Then, power is equally allocated between two
users and the total power of the network is $P$. Denote $y_{ti}$ as
the received signal at time slot $t$ and Antenna $i$. The receiver
stacks the receive signal into a $2N\times 1$ vector as
$\bar{\mb{y}}=\left[y_{11}\ -y_{21}^*\ \cdots \ y_{1N}\
-y_{2N}^*\right]^\t$. The equivalent system equation of
$\bar{\mb{y}}$ can be expressed as
\begin{align}
\bar{\mb{y}}=\sqrt{\frac{P}{4}}\underset{\tilde{\mb{G}}^{(1)}}{\underbrace{\left[\begin{array}{c} \tilde{\mb{G}}_1^{(1)}\\ \vdots \\ \tilde{\mb{G}}_N^{(1)}\end{array}\right]}}\left[\begin{array}{c} s_{1}^{(1)}\\
s_2^{(1)}\end{array}\right]+\sqrt{\frac{P}{4}}\underset{\tilde{\mb{G}}^{(2)}}{\underbrace{\left[\begin{array}{c} \tilde{\mb{G}}_1^{(2)}\\ \vdots \\ \tilde{\mb{G}}_N^{(2)}\end{array}\right]}}\left[\begin{array}{c} s_{1}^{(2)}\\
s_{2}^{(2)}\end{array}\right]+\bar{\mb{n}},
\end{align}
where $\bar{\mb{n}}=\left[n_{11}\ -n_{21}^*\ \cdots \ n_{1N}\
-n_{2N}^*\right]^\t$ denotes the equivalent $2N\times 1$ noise
vector and $\tilde{\mb{G}}_i^{(k)}$ denotes the equivalent Alamouti
channel matrix from User $k$ to Antenna $i$ as that defined in
\eqref{eq-alarec}.

The receiver decouples four symbols through two steps. In Step
1, symbols of each user are separated using ZF. Two
$2(N-1)\times 2N$ ZF matrices are formed as
\begin{align}\label{eq-usersep}
\bar{\mb{Z}}^{(k)}=\left[\begin{array}{ccccc}\frac{2\tilde{\mb{G}}_1^{(k)*}}{\|\tilde{\mb{G}}_1^{(k)}\|^2}&
\mb{0}_2&\cdots & \mb{0}_2 &
-\frac{2\tilde{\mb{G}}_N^{(k)*}}{\|\tilde{\mb{G}}_N^{(k)}\|^2}\\
\mb{0}_2 &
\frac{2\tilde{\mb{G}}_2^{(k)*}}{\|\tilde{\mb{G}}_2^{(k)}\|^2}&
\cdots & \mb{0}_2 &
-\frac{2\tilde{\mb{G}}_N^{(k)*}}{\|\tilde{\mb{G}}_N^{(k)}\|^2}\\
\vdots & \vdots & \ddots &\vdots&\vdots\\
\mb{0}_2&\mb{0}_2&\cdots&
\frac{2\tilde{\mb{G}}_{N-1}^{(k)*}}{\|\tilde{\mb{G}}_{N-1}^{(k)}\|^2}&
-\frac{2\tilde{\mb{G}}_N^{(k)*}}{\|\tilde{\mb{G}}_N^{(k)}\|^2}
\end{array}\right],\ k\in \{1,2\},
\end{align}
where the receiving filter $\bar{\mb{Z}}^{(k)}$ is used to
zero-force User $k$'s symbols, i.e.
$\bar{\mb{Z}}^{(k)}\tilde{\mb{G}}^{(k)}=\mb{0}$, and extracts the
other user's symbols. The IC process can be conducted as
\begin{align}
\tilde{\mb{y}}^{(k)}=\bar{\mb{Z}}^{(\underline{k})}\bar{\mb{y}}=\sqrt{\frac{P}{4}}\bar{\mb{Z}}^{(\underline{k})}\tilde{\mb{G}}^{(k)}\left[\begin{array}{c} s_{1}^{(k)}\\
s_2^{(k)}\end{array}\right]+\bar{\mb{Z}}^{(\underline
k)}\bar{\mb{n}},
\end{align}
where the index $\underline{k}=\{1,2\}\setminus{k}$, denotes
the user other than User $k$. The vector
$\tilde{\mb{y}}^{(k)}$ contains only the symbols of User $k$.
In Step 2, two symbols from the same user are further
decoupled. Note that the $2\times 2$ submatrices contained in
the resulting equivalent channels
$\bar{\mb{Z}}^{(\underline{k})}\tilde{\mb{G}}^{(k)}$ have
Alamouti structures due to the completeness of Alamouti
matrices under addition and multiplication. The receiver
constructs the $2\times 2(N-1)$ symbol separating filters for
User $k$ as
\begin{align}\label{eq-symsep}
\mb{F}^{(k)}=\alpha^{(k)}\left(\bar{\mb{Z}}^{(\underline
k)}\tilde{\mb{G}}^{(k)}\right)^*,
\end{align}
where the coefficient $\alpha^{(k)}$ normalizes the equivalent
combined filter, i.e.,
$\frac{\|\mb{F}^{(k)}\bar{\mb{Z}}^{(\underline{k})*}\|^2}{2}=1$.
Therefore, we have $
\alpha^{(k)}=\frac{\sqrt{2}}{\|\tilde{\mb{G}}^{(k)*}\bar{\mb{Z}}^{(\underline
k)*}\bar{\mb{Z}}^{(\underline{k})}\|}$. Two symbols are
separated through
\begin{align}
\mb{F}^{(k)}\tilde{\mb{y}}^{(k)}=\frac{\alpha^{(k)}}{2}\sqrt{\frac{P}{4}}\|\bar{\mb{Z}}^{(\underline{k})}\tilde{\mb{G}}^{(k)}\|^2\left[\begin{array}{c} s_{1}^{(k)}\\
s_2^{(k)}\end{array}\right]+\mb{F}^{(k)}\bar{\mb{Z}}^{(\underline
k)}\bar{\mb{n}},
\end{align}
It can be observed that the resulting channel for each symbol is a
scalar, and each entry in the equivalent noise vector
$\mb{F}^{(k)}\bar{\mb{Z}}^{(\underline{k})*}\bar{\mb{n}}$ can be
verified to be i.~i.~d.~$\mc{CN}(0,1)$. Thus, symbol-by-symbol
decoding can be conducted based on the $l$th entry of
$\mb{F}^{(k)}\tilde{\mb{y}}^{(k)}$ to recover $s_l^{(k)}$. In total,
four symbol-by-symbol decoding procedures are needed at the receiver
to recover the transmitted symbols.

\subsection{The New Downlink IC Scheme}\label{subsec-BC}
In this subsection, we propose our downlink IC scheme, which is the
dual system of the uplink IC scheme. Since the details of the
transform from the original system to the dual system are involved
and similar to that in Subsection \ref{subsec-Alamouti}, we omit the
construction steps, and directly present the new transmission
scheme.

The system diagram is shown in Fig.~\ref{fig-IC-downlink}. In
the downlink system, we use the user indices to denote
receivers rather than transmitters. The transmitter sends two
symbols $s_1^{(k)}$ and $s_2^{(k)}$ to User $k$, encoded
through Alamouti codes. The process of precoding has two
steps. The first precoder $\mb{E}^{(k)}$ is called \emph{the
symbol separating precoder}, and reuses the structure of
symbol separating filters in \eqref{eq-symsep}. Note that
$\mb{F}^{(k)}$ is a function of $g_{ji}^{(k)}$ and
$g_{ji}^{(\underline{k})}$. The $2\times 2(N-1)$ filter
$\mb{E}^{(k)}$ is obtained by replacing every
$g_{ji}^{(\kappa)}$ in $\mb{F}^{(k)}$ with $h_{ij}^{(\kappa)}$
for $\kappa\in\{1,2\}$. The output from the symbol separating
precoder can be written as
\begin{align}\label{eq-sympre}
\tilde{\mb{X}}^{(k)}=\left[\begin{array}{cc} s_1^{(k)}& s_2^{(k)}\\
-s_2^{(k)*}& s_1^{(k)*}\end{array}\right]\mb{E}^{(k)},\ k\in
\{1,2\}.
\end{align}
It allows each receiver to decouple its own symbols without
knowing CSIR, as will be shown later. The second precoder is
called \emph{the user separating precoder}. Let
$\mb{h}_i^{(k)}=\left[h_{i1}^{(k)}\ h_{i2}^{(k)}\right]$,
denoting the $1\times 2$ channel vector from transmit Antenna
$i$ to User $k$. We design the
$2(N-1)\times N$ user separating precoder for User $k$ as %Let $\hat{\mb{G}}$ consist of the
%first and third columns of $\bar{\mb{Z}}^{(1)}$ and $\hat{\mb{H}}$
%consist of the first and third columns of $\bar{\mb{Z}}^{(2)}$
%defined in \eqref{eq-usersep}.
\begin{align}\label{eq-ICpre}
\bar{\mb{B}}^{(\underline{k})}=\left[\begin{array}{ccccc}\frac{\mb{h}_{1}^{(\underline{k})*}}{\|\mb{h}_{1}^{(\underline{k})}\|^2}&\mb{0}_{2,1}&\cdots&\mb{0}_{2,1}&-\frac{\mb{h}_{N}^{(\underline{k})*}}{\|\mb{h}_{N}^{(\underline{k})}\|^2}\\
\mb{0}_{2,1}&\frac{\mb{h}_{2}^{(\underline
k)*}}{\|\mb{h}_{2}^{(\underline
k)}\|^2}&\cdots&\mb{0}_{2,1}&-\frac{\mb{h}_{N}^{(\underline
k)*}}{\|\mb{h}_{N}^{(\underline
k)}\|^2}\\ \vdots & \vdots & \ddots & \vdots \\
\mb{0}_{2,1}&\mb{0}_{2,1}&\cdots&\frac{\mb{h}_{N-1}^{(\underline
k)*}}{\|\mb{h}_{N-1}^{(\underline
k)}\|^2}&-\frac{\mb{h}_{N}^{(\underline
k)*}}{\|\mb{h}_{N}^{(\underline{k})}\|^2}\end{array}\right].
\end{align}
It can be observed that each user's separating precoder
depends on the channel coefficients to the other receiver. The
second precoder allows the undesired receiver to cancel the
interference. The transmitted $2\times N$ matrix is linearly
generated by multiplying the Alamouti code matrix with the two
precoders and adding up users' signals
\begin{align}\label{eq-tranX0}
\mb{X}=\sqrt{\frac{P}{2}}\left(\tilde{\mb{X}}^{(1)}\bar{\mb{B}}^{(2)}+\tilde{\mb{X}}^{(2)}\bar{\mb{B}}^{(1)}\right).
\end{align}
The transmitted power in two time slots is
$\underset{s_l^{(k)}}{\Exp}\tr\left(\mb{X}^*\mb{X}\right)=2P$.
%\begin{align*}
%\underset{s_l^{(k)}}{\Exp}\tr\left(\mb{X}^*\mb{X}\right)=&\frac{P}{4}\frac{\|\bar{\mb{G}}_1\|^2\|\bar{\mb{G}}_2\|^2}{\left(\|\bar{\mb{G}}_1\|^2+\|\bar{\mb{G}}_2\|^2\right)}\underset{s_l^{(1)}}{\Exp}\left(\left|s_1^{(1)}\right|^2+\left|s_2^{(1)}\right|^2\right)\tr\left(\hat{\mb{G}}^*\hat{\mb{G}}\right)\\
%&+\frac{P}{4}\frac{\|\bar{\mb{H}}_1\|^2\|\bar{\mb{H}}_2\|^2}{\left(\|\bar{\mb{H}}_1\|^2+\|\bar{\mb{H}}_2\|^2\right)}\underset{s_l^{(2)}}{\Exp}\left(\left|s_1^{(2)}\right|^2+\left|s_2^{(2)}\right|^2\right)\tr\left(\hat{\mb{H}}^*\hat{\mb{H}}\right)=P+P=2P.
%\end{align*}
Therefore, the system satisfies the short-term power
constraint $P$ per time slot. Also, it can be verified that
\begin{align}\label{eq-powercons}\underset{s_l^{(k)}}{\Exp}\tr\left(\frac{P}{2}\tilde{\mb{X}}^{(k)*}\bar{\mb{B}}^{(\underline{k})}\bar{\mb{B}}^{(\underline
k)*}\tilde{\mb{X}}^{(k)}\right)=P,\ k\in\{1,2\},
\end{align}
i.e., power is equally allocated between the two
users\footnote{Since power is equally allocated between the two
users in the original uplink MAC systems, the dual systems also have
equal power allocation. With CSIT, power allocation is possible, and
will
be discussed later. }. %The receive $2\times 2$ signal matrix at User
%$k$ can be written as
%\begin{align}
%\mb{R}^{(k)}=\mb{X}\mb{H}^{(k)}+\mb{W}^{(k)},
%\end{align}
%where $\mb{W}^{(k)}$ denotes the $2\times 2$ $\mc{CN}(0,1)$ AWGN
%matrix at User $k$. User $k$ decouples symbols and cancels
%interference with one step.

Let us denote the receive signal at time slot $t$ and Antenna
$j$ of User $k$ by $r_{tj}^{(k)}$ and i.~i.~d.~$\mc{CN}(0,1)$
noise by $w_{tj}^{(k)}$. Since channels are Rayleigh flat
fading, we have
\begin{align}
{r}_{tj}^{(k)}=\sum_{i=1:N}{x}_{ti}{h}_{ij}^{(k)}+{w}_{tj}^{(k)},
\end{align}
where ${x}_{ti}$ denotes the $(t,i)$ entry of the transmit block
$\mb{X}$. User $k$ decouples symbols and cancels interference in one
step. Decision variables for $s^{(k)}_1$ and $s^{(k)}_2$ are
constructed as
\begin{align}\label{eq-s1decoup}\bar{r}_1^{(k)}=r_{11}^{(k)}+r_{22}^{(k)*},\\
\label{eq-s2decoup}\bar{r}_2^{(k)}=r_{12}^{(k)}-r_{21}^{(k)*},\end{align}
respectively. The operations are independent of CSIR, and the
receiver does not need to learn the channel information. In
the following proposition, we show that interference is
cancelled and symbols are decoupled.
{\renewcommand{\baselinestretch}{1.4}
\begin{proposition}\label{prop1}
With the operations in \eqref{eq-s1decoup} and \eqref{eq-s2decoup},
the signals of the other user, i.e., User $(\underline{k})$, are
zero-forced. Further, the symbols $s_1^{(k)}$ and $s_2^{(k)}$ are
decoupled.
\end{proposition}
\begin{proof}
See Appendix \ref{app-1}.
\end{proof}}
We emphasize that \eqref{eq-s1decoup} and \eqref{eq-s2decoup} are
similar to \eqref{eq-step1} and \eqref{eq-step2} without noise
normalization. This is because the original systems of both the dual
Alamouti codes and the downlink IC scheme have Alamouti codes at the
transmitter.

From the proof of Proposition \ref{prop1}, we can rewrite the
equivalent system equation at User $k$ as
\begin{align}\label{eq-systemeq}
\left[\begin{array}{c}\bar{r}_1^{(k)}\\\bar{r}_{2}^{(k)}\end{array}\right]^\t=\sqrt{\frac{P}{2}}\frac{\beta^{(k)}}{2}\left[\begin{array}{c}s_1^{(k)}\\
s_2^{(k)}\end{array}\right]^\t\|\tilde{\mb{B}}^{(\underline
k)}\tilde{\mb{H}}^{(k)}\|^2+\mb{w}^{(k)},
\end{align}
where the notations of $\beta^{(k)}$,
$\tilde{\mb{B}}^{(\underline k)}$, $\tilde{\mb{H}}^{(k)}$, and
$\mb{w}^{(k)}$ can be found in the appendices in
\eqref{eq-beta}, \eqref{eq-icmatrix}, \eqref{eq-result}, and
\eqref{eq-noise}, respectively. In what follows, we discuss
the decoding of $s_1^{(k)}$ and $s_2^{(k)}$. From
\eqref{eq-noise}, it can be verified that the components in
the equivalent $1\times 2$ noise vector ${\mb{w}}^{(k)}$ are
i.~i.~d.~$\mc{CN}(0,2)$ distributed. The ML decoding for
symbol $s_l^{(k)}$ can be conducted as
\begin{align}
\underset{s^{(k)}_l}{\min}\left\|\bar{r}_l^{(k)}-\sqrt{\frac{P}{2}}\frac{\beta^{(k)}}{2}\|\tilde{\mb{B}}^{(\underline
k)}\tilde{\mb{H}}^{(k)}\|^2s^{(k)}_l\right\|=\underset{s^{(k)}_l}{\max}2\real
\{\bar{r}_l^{(k)*}s^{(k)}_l\}-\sqrt{\frac{P}{2}}\frac{\beta^{(k)}}{2}\|\tilde{\mb{B}}^{(\underline
k)}\tilde{\mb{H}}^{(k)}\|^2\left|s^{(k)}_l\right|^2.
\end{align}

When $s_l^{(k)}$ is drawn from a PSK constellation, the second
term is constant for different points in the constellation.
The ML decoding can be simplified as
$\underset{s^{(k)}_l}{\max}\real
\{\bar{r}_l^{(k)*}s^{(k)}_l\}$. Thus, it can be performed
without knowing CSI. For other constellations, a
decision-feedback method can be used to estimate the
coefficient
$\sqrt{\frac{P}{2}}\frac{\beta^{(k)}}{2}\|\tilde{\mb{B}}^{(\underline
k)}\tilde{\mb{H}}^{(k)}\|^2$ from the previously decoded
symbol. The details of a similar estimation can be found in
\cite{DualCode}.

\subsection{Discussion}\label{subsec-dis}
In this subsection, we discuss power allocation and performance
analysis in terms of the diversity gain and symbol rate.

First, we study power allocation. The transmitted matrix in
\eqref{eq-tranX0} assumes equal power allocation between two users.
With CSIT, we can further allocate power between two users to
compensate the channels in deep fading. We can rewrite
\eqref{eq-tranX0} as
\begin{align}\label{eq-tranX}
\mb{X}=\sqrt{\frac{P}{2}}\left(c_1\tilde{\mb{X}}^{(1)}\bar{\mb{B}}^{(2)}+c_2\tilde{\mb{X}}^{(2)}\bar{\mb{B}}^{(1)}\right),
\end{align}
where $c_k$ denotes the power allocation coefficient of User $k$. To
satisfy the sum power constraint
$\tr\left(\mb{X}^*\mb{X}\right)=2P$, we have $c_1^2+c_2^2=2$. From
\eqref{eq-systemeq}, we can calculate the output SNR for each symbol
as
\begin{align}\label{eq-SNR}
\mr{SNR}^{(k)}=\left(\frac{\sqrt{P}}{4}\beta^{(k)}c_k\left\|\tilde{\mb{B}}^{(\underline
k)}\tilde{\mb{H}}^{(k)}\right\|\right)^2=\frac{Pc_k^2}{8}\underset{b_k}{\underbrace{\tr\left(\tilde{\mb{H}}^{(k)*}\tilde{\mb{B}}^{(\underline
k)*}\left(\tilde{\mb{B}}^{(\underline{k})}\tilde{\mb{B}}^{(\underline
k)*}\right)^{-1}\tilde{\mb{B}}^{(\underline
k)}\tilde{\mb{H}}^{(k)}\right)}}.
\end{align}
To optimize the decoding error probability, the transmitter
distributes the total power to maximize the smaller of
$\mr{SNR}^{(k)}$. The optimization problem can be equivalently
written as
\begin{align}\nonumber
\max_{c_1,c_2}&\ \min\left(\frac{P}{8}c_1^2b_1,
\frac{P}{8}c_2^2b_2\right)\\
\mr{s.t.}&\ c_1^2+c_2^2=2.\label{eq-optimal}
\end{align}
The above problem can be converted to a linear optimization
problem. Using Lagrange multiplier methods, the solution is
$c_k^2=\frac{2b_{\underline{k}}}{b_1+b_2}$. The resulting
output SNR can be calculated as
\begin{align}\mr{SNR}^{(1)}=\mr{SNR}^{(2)}=\frac{Pb_1b_2}{4(b_1+b_2)}.\end{align}

The proposed downlink IC scheme satisfies the short-term power
constraints in \eqref{eq-powercons} and the decoding delay of
two time slots. With perfect CSIT, the short-term behavior of
the decoding error probability for any space-time block coding
system is subject to finite diversity gain\cite{FinDiv}. We
provide the diversity gain analysis in the following theorem.

\begin{theorem}\label{thm}
For a two-user $N\times 2$ BC system with perfect CSIT, no CSIR, and
equal-energy constellations, the downlink IC scheme achieves a
diversity gain of $2(N-1)$ for both the equal power allocation and
the optimal power allocation given in \eqref{eq-optimal}.
\end{theorem}
\begin{proof}
See Appendix \ref{app-2}.
\end{proof}
The IC scheme in the two-user $2\times N$ MAC system can achieve
diversity gain $2(N-1)$\cite{KaJa-2}. This theorem says that the
downlink IC scheme can achieve the same diversity gain as the uplink
IC scheme. Intuitively, this can be predicted from Proposition
\ref{prop-1} because these two dual systems have the same
distributions on the instantaneous receive SNRs with equal power
allocation. Since diversity gain depends on the outage probability
of instantaneous receive SNR\cite{divthm_report}, these two systems
achieve the same diversity gains. Further power allocation for the
downlink scheme cannot degrade the resulting diversity gain. Hence,
the downlink scheme with the optimal power allocation can also
achieve a diversity gain of $2(N-1)$. From Theorem \ref{thm}, we can
also observe that the full receive diversity gain of $2$ is
achieved, and the transmit diversity gain is $N-1$. Compared to the
method concatenating STBCs and BD transmission, with the receive
diversity of $2$ and the transmit diversity of $N-2$\cite{Lee10},
our proposed scheme has a higher transmit diversity at the same
transmission rate.

Finally, we end this section with a discussion of the rate.
Each user equivalently receives an Alamouti code in two time
slots. Thus, the symbol rate is $1$ symbol/channel use/user,
and the throughput for the whole network is $2$
symbols/channel use. When the transmitter is equipped with two
antennas, i.e., $N=2$, our proposed scheme achieves the
maximum multiplexing gain. When the transmitter has more than
two antennas, our scheme does not achieve the maximum
multiplexing gain $\min\{N,4\}$, yet still has a rate benefit
compared to the TDMA orthogonal methods.
%When the constellation of each symbol is
%allowed to expand with SNR, the scheme sends two data streams per
%channel use, and achieves multiplexing gain $2$. Since the maximum
%multiplexing gain for the two-user $N\times 2$ BC can be
%upperbounded by collocating receivers, i.e., $\min\{N,4\}$, the
%proposed scheme achieves the maximum multiplexing gain only when
%$N=2$.

\section{Numerical Results}\label{sec-numer}
We show the simulated BER performance of the proposed dual Alamouti
codes and the downlink IC scheme, and compare their performance with
related schemes in the literature. Figures in this section have the
average receive SNR, measured in dB, as horizontal axis and BER as
vertical axis. Since noises and fading channels are normalized, the
average receive SNR is identical to the transmit power.

First, we compare the dual Alamouti codes with the original
Alamouti codes, differential transmission\cite{TarJa00}, and
singular value decomposition (SVD) transmission using the
largest eigen-direction. All schemes have a rate of one and
achieve full spatial diversity for point-to-point MIMO
systems. But the channel information requirements are
different. The dual Alamouti codes need CSIT but no CSIR,
while the original Alamouti codes need CSIR but no CSIT. For
differential transmission, neither CSIT nor CSIR is required.
The SVD scheme needs both CSIT and CSIR.
Fig.~\ref{fig-MIMOcomp} shows the simulated BER performance in
a $2\times 2$ MIMO system with QPSK and 8PSK constellations.
It can be observed that the dual Alamouti codes achieve the
same BER as the original Alamouti codes. This can be confirmed
by Proposition \ref{prop-1}. The SVD scheme outperforms the
other three transmission schemes due to its use of both CSIT
and CSIR. The performance gap between the SVD scheme and the
dual codes is approximately $2.5$~dB. Compared to the SVD
scheme, the dual codes trade BER performance for less
resources to learn CSIR.

Next, the proposed downlink IC schemes, with both equal and
optimal power allocation, are compared with the BD methods
using STBCs\cite{ChAnHe04,Lee10} and an opportunistic TDMA
scheme. Note that for a two-user $N\times 2$ BC, the BD method
requires that the transmitter has at least four antennas.
Alamouti codes are used on top of the ZF precoder to improve
the diversity gain and reduce decoding
complexity\cite{ChAnHe04}. Such a BD method achieves 1
symbol/channel use/user and has the same symbol rate as the
downlink IC scheme. Global CSIR is required at each receiver
to decouple its two symbols carried in the Alamouti codes for
each transmission. To compare with a TDMA scheme with CSIT and
no CSIR, we propose an opportunistic TDMA scheme that assigns
orthogonal time slots for each user and schedules the user
with the stronger norm of channel coefficients to transmit
using dual Alamouti codes. The opportunistic TDMA scheme has
averagely only $0.5$ symbol/channel use/user, and a
higher-order constellation is needed to compensate the rate
loss. Also, it requires a longer decoding delay compared to
both concurrent transmission schemes.

Figs.~\ref{fig-BCcompR1} and \ref{fig-BCcompR2} exhibit the
BER performance at rate=$1$ bit/channel user/user and rate=$2$
bits/channel use/user, respectively. In
Fig.~\ref{fig-BCcompR1}, BPSK is used for the downlink IC
schemes and the BD methods, and QPSK for the opportunistic
TDMA scheme; in Fig.~\ref{fig-BCcompR2}, the downlink IC
schemes and the BD methods use QPSK modulations, and the
opportunistic TDMA scheme uses a 16-QAM constellation. We
first compare the downlink IC scheme using equal power
allocation with its optimal power allocation. From both
figures, there is approximately $1$~dB array gain improvement
for $N=3$ and $N=4$. The improvement for $N=2$ is about
$0.5$~dB. The observation is consistent with Theorem \ref{thm}
that power allocation improves the array gain but not the
diversity gain. Next, we compare the downlink IC scheme with
the BD method. From Figs.~\ref{fig-BCcompR1} and
\ref{fig-BCcompR2}, the downlink IC scheme with $N=4$
outperforms the BD method with $N=4$ in the entire simulated
SNR regime. Also with $N=3$, the downlink IC scheme
outperforms the BD method with $N=4$ for SNR $>10$~dB. It can
be observed that the proposed downlink IC scheme achieves a
higher diversity gain compared to the BD method.

Finally, we compare the downlink IC scheme with our
opportunistic TDMA scheme. In Fig.~\ref{fig-BCcompR1}, the
opportunistic TDMA scheme has substantial gain over the
downlink IC scheme: the opportunistic TDMA scheme with $N=2$
can outperform the downlink IC scheme with $N=4$. This is
because the opportunistic TDMA scheme exploits the multiuser
diversity gain of $2N\times 2=4N$. It also explains the gains
in Fig.~\ref{fig-BCcompR2} for SNR$>11$ dB. On the other hand,
it can be observed in Fig.~\ref{fig-BCcompR2} that for $N=3$
and $N=4$, the downlink IC scheme has approximately $1$~dB
gain over the opportunistic TDMA scheme for SNR$<10$ dB. The
improvement is because the downlink IC scheme has a higher
symbol rate compared to
the opportunistic TDMA scheme. %These observations imply that for low
%and median SNR ranges, the opportunistic TDMA scheme can be used for
%low-rate transmission, while the downlink IC schemes can be
%performed for high-rate transmission.

\section{Conclusions}\label{sec-conclusion}
This paper investigates designs of communication systems with CSIT
and no CSIR. We show such scenarios arise for concurrent
transmissions in BC systems when users do not know the channels of
other users. A duality principle has been proposed for systems with
ZF designs. The duality principle connects the systems that know
CSIR but not CSIT with the systems that know CSIT but not CSIR. We
show an example to construct the dual system of the Alamouti codes,
and propose the \emph{dual Alamouti codes} for point-to-point MIMO
systems. For the two-user downlink MIMO BC, we consider an IC scheme
for the uplink MAC as the original system, and derive its downlink
dual system, called the \emph{downlink IC scheme}. The transmitter
uses CSIT to design linear precoders and each receiver cancels
interference and decouples its own data streams using two linear
operations independent of CSIR. Power allocation between two users
are also discussed. For a two-user $N\times 2$ BC, the downlink IC
schemes achieve a diversity gain of $2(N-1)$ at rate $1$
symbol/channel use/user with both equal and optimal power
allocations. The proposed schemes trade higher diversity for rate
compared to the full-rate BD scheme. Simulation results demonstrate
their superior BER performance over the BD methods concatenated with
STBCs, which require global CSIR at each receiver.

\section*{Appendices}
\subsection{Proof of Proposition \ref{prop1}}\label{app-1}
We prove this proposition for User $1$ only. Since the network
is symmetrical, i.e., the user indices can be exchanged,
similar results can be obtained for User $2$. We first show
that the signals of User $2$ is zero-forced at User $1$. Since
the $2\times 2(N-1)$ matrix $\mb{E}^{(k)}$ has the same
structure as $\mb{F}^{(k)}$, the $2\times 2$ submatrices of
$\mb{E}^{(k)}$ also have Alamouti structures. Then,
$\tilde{\mb{X}}^{(k)}$ in \eqref{eq-sympre} is composed of
$2\times 2$ matrices with Alamouti structures due to the
completeness of Alamouti structures under multiplication. We
can represent $\tilde{\mb{X}}^{(k)}$ by {\small\begin{align}
\tilde{\mb{X}}^{(k)}=\left[\begin{array}{ccccc}\tilde{x}_{11}^{(k)}&\tilde{x}_{21}^{(k)}&\cdots &\tilde{x}_{1(N-1)}^{(k)}&\tilde{x}_{2(N-1)}^{(k)}\\
-\tilde{x}_{21}^{(k)*}&\tilde{x}_{11}^{(k)*}&\cdots&
-\tilde{x}_{2(N-1)}^{(k)*}&\tilde{x}_{1(N-1)}^{(k)*}
\end{array}\right], k \in \{1,2\}.
\end{align}}
The received signals $r_{11}^{(1)}$ and $r_{22}^{(1)}$ can be
expanded as {\setlength{\arraycolsep}{1pt}\small\begin{align*}
&{r}_{11}^{(1)}=\sqrt{\frac{P}{2}}\sum_{k}\underset{\tilde{\mb{x}}^{(k)}}{\underbrace{\left[\begin{array}{c}\tilde{x}_{11}^{(k)}\\\tilde{x}_{21}^{(k)}\\\vdots\\ \tilde{x}_{1(N-1)}^{(k)}\\\tilde{x}_{2(N-1)}^{(k)} \end{array}\right]^\t}}\left[\begin{array}{cccc}\frac{h_{11}^{(\underline{k})*}}{\|{\mb{h}}_1^{(\underline{k})}\|^2}&0&\cdots&-\frac{h_{N1}^{(\underline{k})*}}{\|{\mb{h}}^{(\underline{k})}_N\|^2}\\
\frac{h_{12}^{(\underline{k})*}}{\|{\mb{h}}_1^{(\underline{k})}\|^2}&0&\cdots&-\frac{h_{N2}^{(\underline{k})*}}{\|{\mb{h}}^{(\underline{k})}_N\|^2}\\ \vdots& \ddots&\ddots & \vdots \\ 0& \cdots & \frac{h_{(N-1)1}^{(\underline{k})*}}{\|{\mb{h}}_{N-1}^{(\underline{k})}\|^2}&-\frac{h_{N1}^{(\underline{k})*}}{\|{\mb{h}}^{(\underline{k})}_N\|^2}\\ 0&\cdots & \frac{h_{(N-1)2}^{(\underline{k})*}}{\|{\mb{h}}_{N-1}^{(\underline{k})}\|^2}&-\frac{h_{N2}^{(\underline{k})*}}{\|{\mb{h}}^{(\underline{k})}_N\|^2}\end{array}\right]\left[\begin{array}{c}{h}_{11}^{(1)}\\
h_{21}^{(1)}\\ \vdots\\
h_{N1}^{(1)}\end{array}\right]+w_{11},\\
%&{r}_{22}^{(1)*}=\sqrt{\frac{P}{2}}\sum_{k}\left[\begin{array}{c}-\tilde{x}_{21}^{(k)}\\\tilde{x}_{11}^{(k)}\\\vdots\\ -\tilde{x}_{2(N-1)}^{(k)}\\\tilde{x}_{1(N-1)}^{(k)} \end{array}\right]^\t\left[\begin{array}{cccc}\frac{h_{11}^{(\underline{k})}}{\|{\mb{h}}_1^{(\underline{k})}\|^2}&0&\cdots&-\frac{h_{N1}^{(\underline{k})}}{\|{\mb{h}}^{(\underline{k})}_N\|^2}\\
%\frac{h_{12}^{(\underline{k})}}{\|{\mb{h}}_1^{(\underline{k})}\|^2}&0&\cdots&-\frac{h_{N2}^{(\underline{k})}}{\|{\mb{h}}^{(\underline{k})}_N\|^2}\\ \vdots& \ddots&\ddots & \vdots \\ 0&\cdots & \frac{h_{(N-1)1}^{(\underline{k})}}{\|{\mb{h}}_{N-1}^{(\underline{k})}\|^2}&-\frac{h_{N1}^{(\underline{k})}}{\|{\mb{h}}^{(\underline{k})}_N\|^2}\\ 0&\cdots & \frac{h_{(N-1)2}^{(\underline{k})}}{\|{\mb{h}}_{N-1}^{(\underline{k})}\|^2}&-\frac{h_{N2}^{(\underline{k})}}{\|{\mb{h}}^{(\underline{k})}_N\|^2}\end{array}\right]\left[\begin{array}{c}{h}_{12}^{(1)*}\\
%h_{22}^{(1)*}\\ \vdots\\ h_{N2}^{(1)*}\end{array}\right]+w_{22}^*\\
&{r}_{22}^{(1)*}=\sqrt{\frac{P}{2}}\sum_{k}\tilde{\mb{x}}^{(k)}\left[\begin{array}{cccc}\frac{h_{12}^{(\underline{k})}}{\|{\mb{h}}_1^{(\underline{k})}\|^2}&0&\cdots&-\frac{h_{N2}^{(\underline{k})}}{\|{\mb{h}}^{(\underline{k})}_N\|^2}\\
-\frac{h_{11}^{(\underline{k})}}{\|{\mb{h}}_1^{(\underline{k})}\|^2}&0&\cdots&\frac{h_{N1}^{(\underline{k})}}{\|{\mb{h}}^{(\underline{k})}_N\|^2}\\ \vdots& \ddots&\ddots & \vdots \\ 0&\cdots & \frac{h_{(N-1)2}^{(\underline{k})}}{\|{\mb{h}}_{N-1}^{(\underline{k})}\|^2}&-\frac{h_{N2}^{(\underline{k})}}{\|{\mb{h}}^{(\underline{k})}_N\|^2}\\ 0&\cdots & -\frac{h_{(N-1)1}^{(\underline{k})}}{\|{\mb{h}}_{N-1}^{(\underline{k})}\|^2}&\frac{h_{N1}^{(\underline{k})}}{\|{\mb{h}}^{(\underline{k})}_N\|^2}\end{array}\right]\left[\begin{array}{c}{h}_{12}^{(1)*}\\
h_{22}^{(1)*}\\ \vdots\\ h_{N2}^{(1)*}\end{array}\right]+w_{22}^*.
\end{align*}}
Denote $\bar{\mb{h}}^{(1)}=\left[ h_{11}^{(1)}\ -h_{12}^{(1)*}\
\cdots\ h_{N1}^{(1)}\ -h_{N2}^{(1)*} \right]^\t$ and
{\setlength{\arraycolsep}{2pt}\begin{align}\label{eq-ala2}
\tilde{\mb{H}}_i^{(\underline
k)}=\left[\begin{array}{cc}h_{i1}^{(\underline{k})}& h_{i2}^{(\underline{k})}\\
-h_{i2}^{(\underline{k})*} & h_{i1}^{(\underline
k)*}\end{array}\right],
\end{align}}
which has an Alamouti structure. It follows that
\begin{align}\label{eq-icmatrix}
\bar{r}_1^{(1)}=r_{11}^{(1)}+r_{22}^{(1)*}=\sqrt{\frac{P}{2}}\sum_{k}\tilde{\mb{x}}^{(k)}\tilde{\mb{B}}^{(\underline{k})}
%\underset{}{\underbrace{\left[\begin{array}{cccc}\frac{\tilde{\mb{H}}_1^{(\underline{k})*}}{\|{\mb{h}}_1^{(\underline{k})}\|^2}&\mb{0}_2&\cdots& -\frac{\tilde{\mb{H}}_N^{(\underline{k})*}}{\|{\mb{h}}_N^{(\underline{k})}\|^2}\\
%\vdots&\ddots &\vdots &\vdots\\
%\mb{0}_2&\cdots&\frac{\tilde{\mb{H}}_{N-1}^{(\underline
%k)*}}{\|{\mb{h}}_{N-1}^{(\underline
%k)}\|^2}&-\frac{\tilde{\mb{H}}_N^{(\underline
%k)*}}{\|{\mb{h}}_N^{(\underline
%k)}\|^2}\end{array}\right]}}
\bar{\mb{h}}^{(1)}+w_{11}^{(1)}+w_{22}^{(1)*},
\end{align}
where the $2(N-1)\times 2N$ matrix
$\tilde{\mb{B}}^{(\underline{k})}$ has the same structure as
the ZF matrix $\bar{\mb{Z}}^{(k)}$ in \eqref{eq-usersep} and
replaces each $2\times 2$ matrix
$\frac{2\tilde{\mb{G}}_i^{(k)*}}{\|\tilde{\mb{G}}_i^{(k)}\|^2}$
with
$\frac{\tilde{\mb{H}}_i^{(\underline{k})}}{\|\mb{h}_{i}^{(\underline
k)}\|^2}$ for $i=1,\ldots,N$. Further let
{\small$\tilde{\mb{h}}^{(1)}=\left[ h_{12}^{(1)}\
h_{11}^{(1)*}\ \cdots\ h_{N2}^{(1)} \ h_{N1}^{(1)*}
\right]^\t$}. Similarly, we have
{\small$\bar{r}_{2}^{(1)}=r_{12}^{(1)}-r_{21}^{(1)*}=\sqrt{\frac{P}{2}}\underset{k}{\sum}\tilde{\mb{x}}^{(k)}\tilde{\mb{B}}^{(\underline
k)}\tilde{\mb{h}}^{(1)}+w_{12}^{(1)}-w_{21}^{(1)*}. $} Combine
these two equations as,
{\small\setlength{\arraycolsep}{1pt}\begin{align}
\left[\begin{array}{cc}\bar{r}_1^{(1)}&\bar{r}_{2}^{(1)}\end{array}\right]&=\sqrt{\frac{P}{2}}\sum_{k}\tilde{\mb{x}}^{(k)}\tilde{\mb{B}}^{(\underline{k})}\underset{\tilde{\mb{H}}^{(1)}}{\underbrace{\left[\bar{\mb{h}}^{(1)}\
\tilde{\mb{h}}^{(1)}\right]}}+\mb{w}^{(1)}\label{eq-result}
=\sqrt{\frac{P}{2}}\tilde{\mb{x}}^{(1)}\tilde{\mb{B}}^{(2)}\tilde{\mb{H}}^{(1)}+\mb{w}^{(1)},
\end{align}}
where $\mb{w}^{(1)}$ denotes the equivalent $1\times 2$ noise
vector {\small\begin{align}\label{eq-noise}
\mb{w}^{(1)}=\left[\begin{array}{cc}w_{11}^{(1)}+w_{22}^{(1)*}&
w_{12}^{(1)}-w_{21}^{(1)*}
\end{array}\right].
\end{align}}
Note that
{\small$\tilde{\mb{H}}^{(1)}=\left[\tilde{\mb{H}}_1^{(1)*}\
\cdots \ \tilde{\mb{H}}_N^{(1)*}\right]^*$}, whose $2\times 2$
submatrice denotes the equivalent Alamouti channel matrix from
Antenna $i$. The signal of User $1$ is cancelled because
$\tilde{\mb{B}}^{(2)}$ zero-forces $\tilde{\mb{H}}^{(1)}$.

In what follows, we show that the symbols $s_1^{(1)}$ and
$s_2^{(1)}$ are decoupled. From the structure of the symbol
separating filter, we have $\tilde{\mb{x}}^{(1)}=\left[s_1^{(1)}\
s_2^{(1)}\right]\mb{E}^{(1)}$. Note that $\mb{E}^{(1)}$ has the same
structure as $\mb{F}^{(1)}$ but replacing $g_{ji}^{(\kappa)}$ with
$h_{ij}^{(\kappa)}$ for $\kappa \in \{1,2\}$. From
\eqref{eq-symsep}, we have
\begin{align}\label{eq-beta}\mb{E}^{(1)}=\beta^{(1)}\left(\tilde{\mb{B}}^{(2)}\tilde{\mb{H}}^{(1)}\right)^*\end{align}
where
$\beta^{(1)}=\frac{\sqrt{2}}{\|\tilde{\mb{H}}^{(1)*}\tilde{\mb{B}}^{(2)*}\tilde{\mb{B}}^{(2)}\|}$.
Eq.~\eqref{eq-result} can be further written as
{\setlength{\arraycolsep}{1pt}\small\begin{align*}
\left[\begin{array}{c}\bar{r}_1^{(1)}\\\bar{r}_{2}^{(1)}\end{array}\right]^\t&
%=\sqrt{\frac{P}{2}}\left[\begin{array}{c}s_1^{(1)}\\
%s_2^{(1)}\end{array}\right]^\t\mb{E}^{(1)}\tilde{\mb{B}}^{(2)}\tilde{\mb{H}}^{(1)}+\mb{w}^{(1)}=
=\sqrt{\frac{P}{2}}\left[\begin{array}{c}s_1^{(1)}\\
s_2^{(1)}\end{array}\right]^\t\beta^{(1)}\left(\tilde{\mb{B}}^{(2)}\tilde{\mb{H}}^{(1)}\right)^*\tilde{\mb{B}}^{(2)}\tilde{\mb{H}}^{(1)}+\mb{w}^{(1)}=\sqrt{\frac{P}{2}}\frac{\beta^{(1)}}{2}\left[\begin{array}{c}s_1^{(1)}\\
s_2^{(1)}\end{array}\right]^\t\|\tilde{\mb{B}}^{(2)}\tilde{\mb{H}}^{(1)}\|^2+\mb{w}^{(1)}.
\end{align*}}
The second equality is achieved because the $2\times 2$
submatrices of $\tilde{\mb{B}}^{(2)}\tilde{\mb{H}}^{(1)}$ have
Alamouti structures.

\subsection{Proof of Theorem \ref{thm}}\label{app-2}
When the system uses an equal-energy constellation, the ML
decoding can be conducted without CSIR. Then, the diversity
gain performance depends only on the output instantaneous
receive SNR. A method analyzing the diversity gain based on
instantaneous receive SNR is presented in
\cite{divthm_report}. The techniques in \cite{divthm_report}
focus on the first-order exponent of the outage probability of
the instantaneous receive SNR. For equal power allocation with
$c_k=1$, the resulting SNR can be rewritten as
$\mr{SNR}^{(k)}=\frac{Pb_k}{8}$ from \eqref{eq-SNR}. Note that
$b_k$ has the same expression as the instantaneous normalized
receive SNR of the uplink IC scheme that achieves a diversity
gain of $2(N-1)$\cite{KaJa-2}. Thus, with equal power
allocation, the downlink IC scheme also achieves a diversity
gain of $2(N-1)$.

When power is distributed to maximize the smaller of the
output SNRs at each user, the resulting SNR can be expressed
as $\mr{SNR}^{(k)}=\frac{Pb_1b_2}{4(b_1+b_2)}$. From the
definition of the optimization problem in \eqref{eq-optimal},
we always have $\frac{Pb_1b_2}{4(b_1+b_2)}\ge
\min\left(\frac{Pb_1}{8}, \frac{Pb_2}{8}\right)$. Further,
because of the inequality $\frac{b_1b_2}{b_1+b_2}\le
\frac{b_1b_2}{b_k}=b_{\underline{k}}$, we can obtain
$\frac{Pb_1b_2}{4(b_1+b_2)}\le \frac{P}{4}\min\left(b_1,
b_2\right)$. It follows that $ \frac{P}{8}\min\left(b_1,
b_2\right)\le \mr{SNR}^{(k)}\le \frac{P}{4}\min\left(b_1,
b_2\right). $ This implies that $\mr{SNR}^{(k)}$ scales with
$P\min\left(b_1, b_2\right)$. Note that both $b_1$ and $b_2$
have diversity gain $2(N-1)$. Then, $\min\left(b_1,
b_2\right)$ also has diversity gain $2(N-1)$. Therefore, the
system with optimal power allocation achieves diversity gain
$2(N-1)$.

\renewcommand{\baselinestretch}{1.1}
\footnotesize
\bibliographystyle{ieeetran}
\bibliography{IEEEabrv,IA}

\newpage
\renewcommand{\baselinestretch}{1.6}

\begin{table}
  \centering
  \caption{Classification of multi-antenna communication systems with respect to the channel information}\label{table}
  \begin{tabular}{|c|c|c|c|}
  \hline
  Category & CSIT & CSIR & References \\
  \hline
  System A & No & No & \cite{Hugh00,HoSw00,TarJa00} \\
  System B & No & Yes & \cite{Alamouti98,tar99} \\
  System C & Yes & Yes & \cite{JSO02,ZhGi02} \\
  System D & Yes & No &  \\
  \hline
\end{tabular}
\end{table}

\begin{figure}
  \centering
  \includegraphics[width=5in]{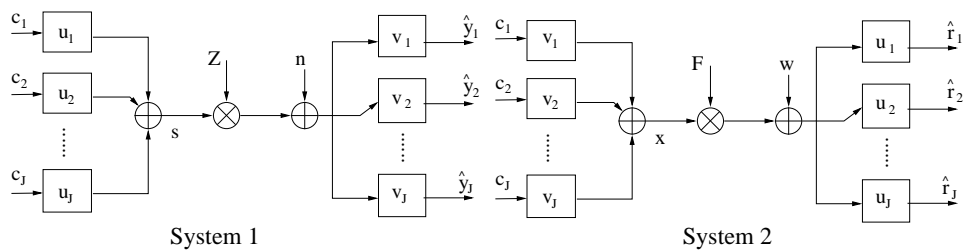}\\
  \caption{System diagram for two real systems. They are dual to each other under ZF designs.}\label{fig-dual}
\end{figure}

\begin{figure}
  \centering
  \includegraphics[width=5in]{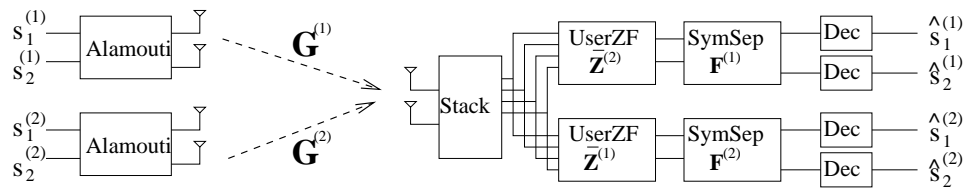}\\
  \caption{Interference cancellation for the two-user uplink MAC.}\label{fig-IC-uplink}
\end{figure}

\begin{figure}
  \centering
  \includegraphics[width=5in]{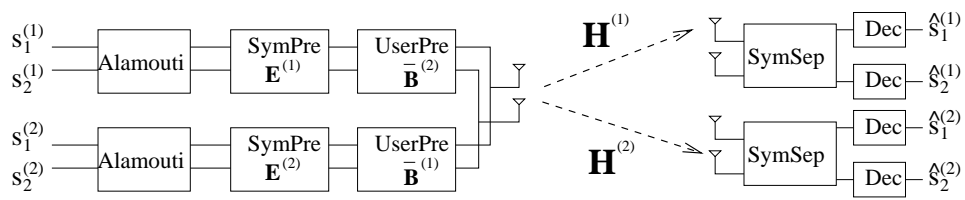}\\
  \caption{Interference cancellation for the two-user downlink BC.}\label{fig-IC-downlink}
\end{figure}

\begin{figure}
  \centering
  \includegraphics[width=4in,angle=-90]{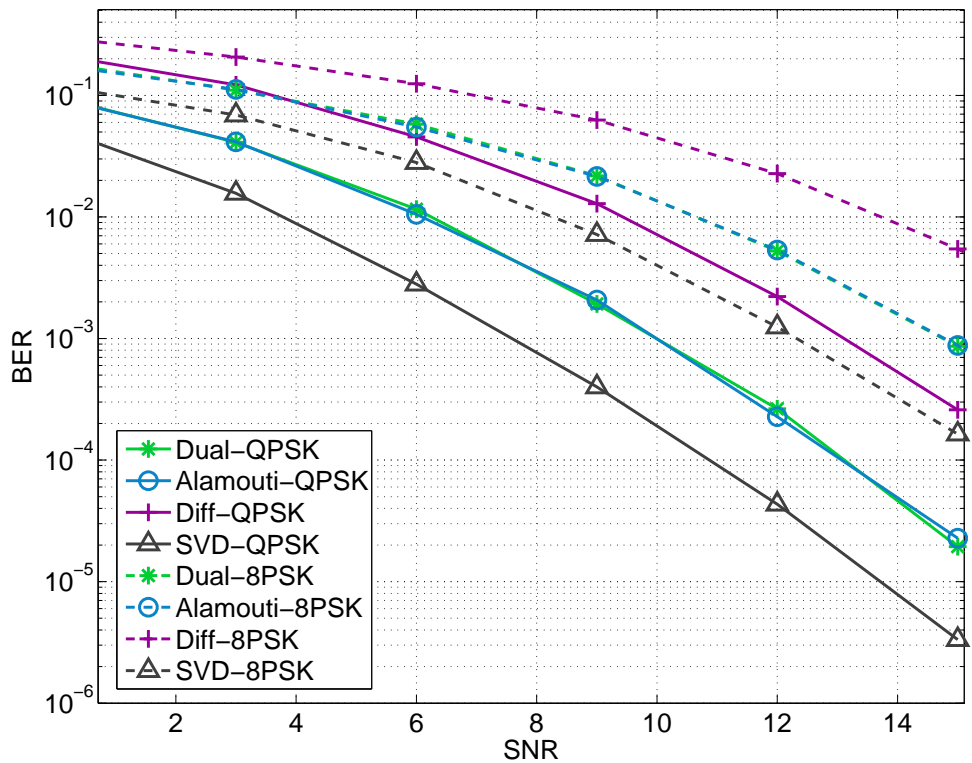}\\
  \caption{Performance comparison in $2\times 2$ MIMO systems: Dual Alamouti codes, Alamouti codes, differential transmission, and SVD methods.}\label{fig-MIMOcomp}
\end{figure}

\begin{figure}
  \centering
  \includegraphics[width=4in,angle=-90]{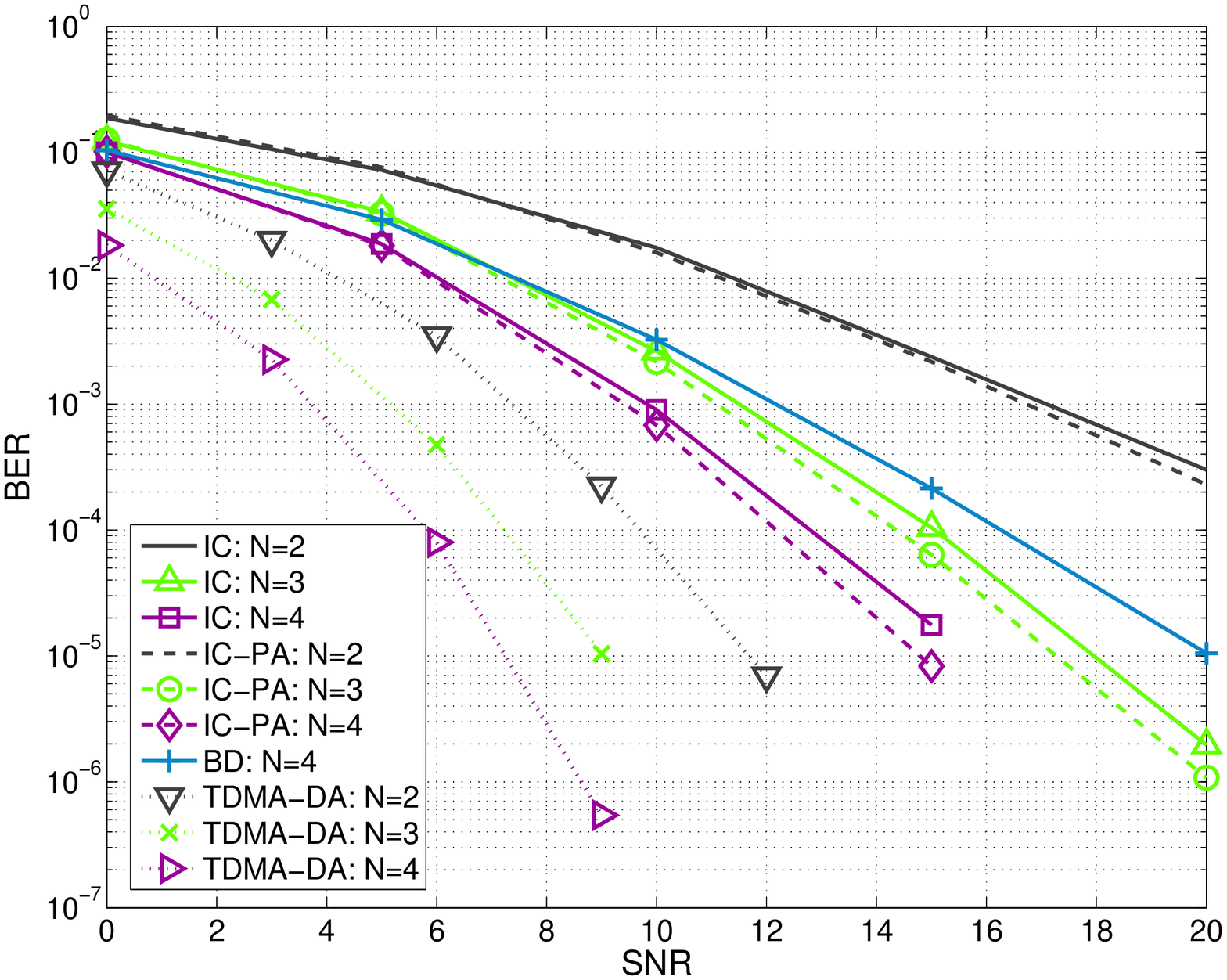}\\
  \caption{Performance comparison in two-user $N\times 2$ BC systems at $1$ bit/channel use/user: Downlink IC scheme (labeled as `IC'), Downlink IC scheme using optimal power allocation (labeled as `IC-PA'), BD methods, and opportunistic TDMA using dual Alamouti codes (labeled as `TDMA-DA').}\label{fig-BCcompR1}
\end{figure}

\begin{figure}
  \centering
  \includegraphics[width=4in,angle=-90]{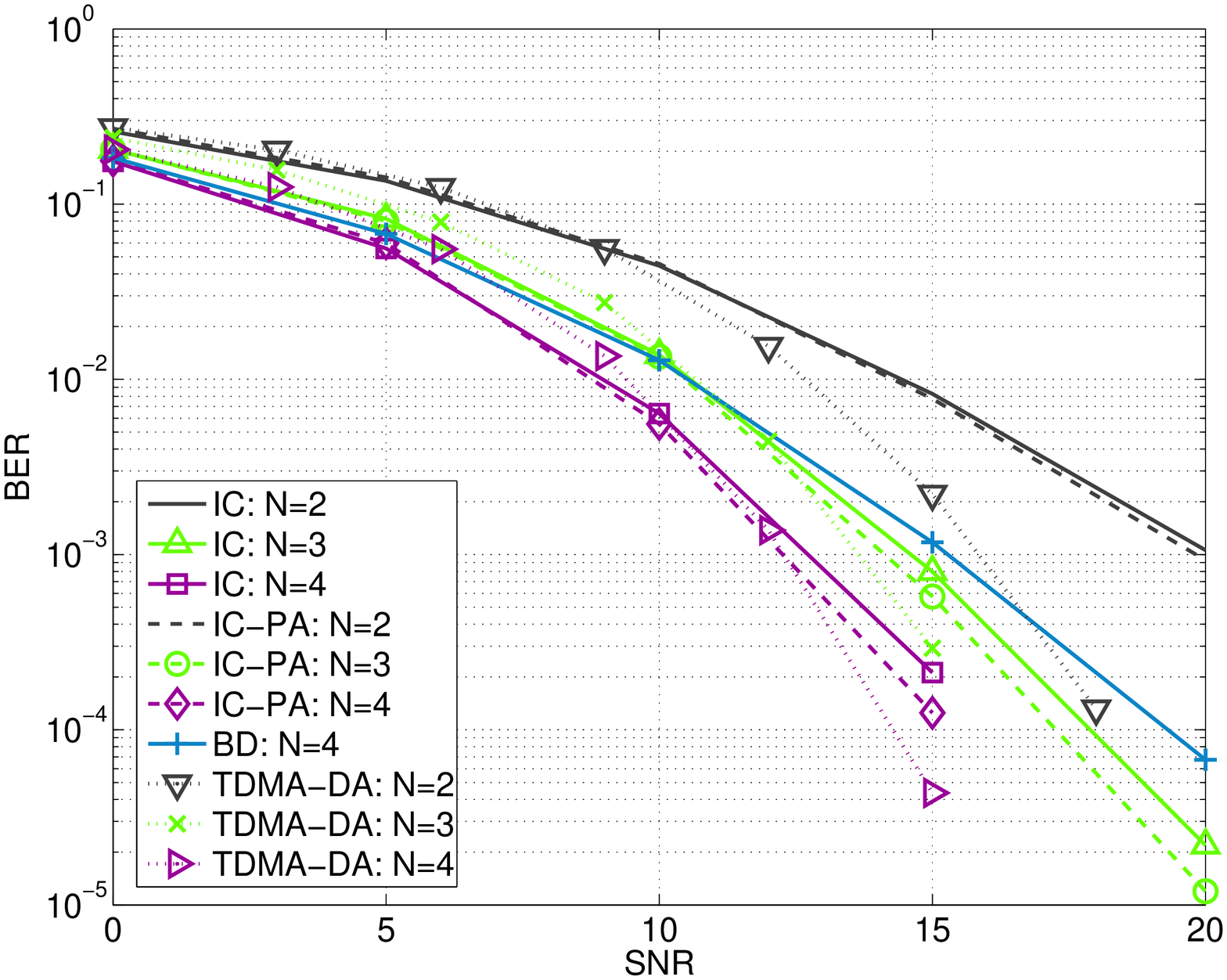}\\
  \caption{Performance comparison in two-user $N\times 2$ MIMO systems at $2$ bits/channel use/user: Downlink IC scheme (labeled as `IC'), Downlink IC scheme using optimal power allocation (labeled as `IC-PA'), BD methods, and opportunistic TDMA using dual Alamouti codes (labeled as `TDMA-DA').}\label{fig-BCcompR2}
\end{figure}

\end{document}